\documentclass{article}

\usepackage[latin1]{inputenc}
\usepackage[T1]{fontenc} 
\usepackage[english]{babel}
\usepackage{amsmath,amssymb,latexsym,eufrak,euscript}
\usepackage{subfigure,pstricks,pst-node,pst-coil}

\usepackage{url,tikz}
\usepackage{pgflibrarysnakes}

\usepackage{multicol}

\usetikzlibrary{arrows}
\usetikzlibrary{automata}
\usetikzlibrary{snakes}
\usetikzlibrary{shapes}

\newtheorem{theorem}{Theorem}
 
\newtheorem{lemma}[theorem]{Lemma}

\newtheorem{construction}[theorem]{Construction}
\newtheorem{corollary}[theorem]{Corollary}

\newenvironment{proof}
{ \noindent {\sc Proof.\/}  }
{\null  \hfill $\Box$ \par\medskip \vspace{1em}}

\newcommand {\pb}[3] {
\medskip
\noindent
\fbox{
\parbox[c]{11.5cm}{
{\bf #1}

\noindent
{\bf Input:} #2

\noindent
{\bf Output:} #3

\medskip
}
}

\medskip
}

\def \A {\mathcal{A}}
\def \P {\mathcal{P}}
\def \B {\mathcal{B}}
\def \C {\mathcal{C}}
\def \D {\mathcal{D}}

\def \T {\mathfrak{T}}
\def \X {\mathcal{X}}

\def \pos {\mathrm{Pos}}

\begin{document}
\label{firstpage}

  \title{The Emptiness Problem for Tree Automata with at
    Least One  Disequality Constraint is NP-hard} 
\author{P.-C. Héam\footnote{FEMTO-ST - INRIA - CNRS - Université de
    Franche-Comté}\and V. Hugot\footnote{LIFL - INRIA } \and  O.
  Kouchnarenko\footnote{FEMTO-ST - INRIA - CNRS -  Université de Franche-Comté}}

\maketitle

\begin{abstract}
The model of tree automata with equality and disequality constraints was
introduced in 2007 by Filiot, Talbot and Tison. In this paper we show that if
there is at least one disequality constraint, the emptiness problem is NP-hard.
\end{abstract}

\section{Introduction}

Tree automata are a pervasive tool of contemporary Computer Science, with
applications running the gamut from XML processing~\cite{xml} to program
verification~\cite{DBLP:conf/rta/BoichutGJR07,DBLP:journals/iandc/JacquemardKV11,DBLP:conf/wia/HeamHK12}. 
Since their original introduction, they have spawned an
ever-growing family of variants, each with its own characteristics of
expressiveness and decision complexity. Among them is the family of tree
automata with {\em equality and disequality constraints}, providing several means for comparing subtrees.
Examples of such automata are the original class introduced in \cite{DBLP:conf/fct/DauchetM79}, 
their restriction to constraints between {\em brothers}~\cite{DBLP:conf/stacs/BogaertT92}, 
and {\em visibly tree automata} with memory and
constraints~\cite{DBLP:journals/lmcs/Comon-LundhJP08}. 
In this paper we
focus on a recently introduced variant: tree automata with
\emph{global} equality and disequality
constraints~\cite{DBLP:conf/csl/FiliotTT07,DBLP:conf/dlt/FiliotTT08,DBLP:journals/ijfcs/FiliotTT10}. For this class of
automata, the universality problem is
undecidable~\cite{DBLP:journals/ijfcs/FiliotTT10}, while  membership
is NP-complete~\cite{DBLP:journals/ijfcs/FiliotTT10}, and 
emptiness is decidable~\cite{DBLP:conf/lics/BargunoCGJV10}. Several
complexity results for subclasses were pointed out in the literature: the
membership problem is polynomial for rigid tree
automata~\cite{DBLP:journals/iandc/JacquemardKV11} as well as for tree
automata with a fixed number of equality
constraints~\cite{DBLP:conf/wia/HeamHK12} and no disequality constraints. The emptiness problem is
EXPTIME-complete if there are only equality
constraints~\cite{DBLP:journals/ijfcs/FiliotTT10}, in NEXPTIME 
if there
are only irreflexive disequality
constraints~\cite{DBLP:journals/ijfcs/FiliotTT10}, and in 3-EXPTIME if there
are only reflexive disequality constraints~\cite{DBLP:journals/jar/CreusGG13}.
In this paper we show that the emptiness problem is NP-hard for tree
automata with global equality and disequality constraints if there is at
least one disequality constraint. 


\section{Formal Background}

A {\it ranked alphabet} is a finite set $\mathcal{F}$ of symbols equipped with an
arity function $\mathsf{arity}$ from $\mathcal{F}$ into $\mathbb{N}$. The set of
{\it terms} on $\mathcal{F}$, denoted $\T(\mathcal{F})$ is inductively defined as the
smallest set satisfying: for every $t\in \mathcal{F}$ such that
$\mathsf{arity}(t)=0$, $t\in \T(\mathcal{F})$; if $t_1,\ldots,t_n$ are in
$\T(\mathcal{F})$ and if $f\in\mathcal{F}$ has arity $n$, then
$f(t_1,\ldots,t_n)\in\T(\mathcal{F})$. The set of {\it positions} of a term $t$,
denoted $\pos(t)$, is the subset of $\mathbb{N}^*$ (finite words over
$\mathbb{N}$) inductively defined by: if $\mathsf{arity}(t)=0$, then
$\pos(t)=\{\varepsilon\}$; if $t=f(t_1,\ldots,t_n)$, where $n$ is the arity
of $f$, then
$\pos(t)=\{\varepsilon\}\cup\{i\cdot \alpha_i\mid \alpha_i\in\pos(t_i)\}$. A
term $t$ induces a function (also denoted $t$) from $\pos(t)$ into $\mathcal{F}$,
where $t(\alpha)$ is the symbol of $\mathcal{F}$ occurring in $t$ at the position
$\alpha$. The subterm of a term $t$ at position $\alpha\in \pos(t)$ is the
term $t_{|\alpha}$ such that
$\pos(t_{|\alpha})=\{\beta\mid \alpha\cdot\beta\in \pos(t)\}$ and for all
$\beta\in \pos(t_{\alpha})$, $t_{|\alpha}(\beta)=t(\alpha\cdot\beta)$. For
any pair of terms $t$ and $t^\prime$, any $\alpha\in\pos(t)$, the term
$t[t^\prime]_{\alpha}$ is the term obtained by substituting in $t$ the
subterm rooted at position $\alpha$ by $t^\prime$. Let $\X$ be an
infinite countable set of variables such that $\X\cap\mathcal{F}=\emptyset$. A
{\it context} $C$ is term in $\T(\mathcal{F}\cup\X)$ (variables are constants) where
each variable occurs at most once; it is denoted $C[X_1,\ldots,X_n]$ if the
occurring variables are $X_1,\ldots,X_n$. If $t_1,\ldots,t_n$ are in
$\T(\mathcal{F})$, $C[t_1,\ldots,t_n]$ is the term obtained from $C$ by
substituting each $X_i$ by $t_i$.

A {\it tree automaton} on a ranked alphabet $\mathcal{F}$ is a tuple
$\A=(Q,\Delta,F)$, where $Q$ is a finite set of states, $F\subseteq Q$ is
the set of final sets and $\Delta$ is a finite set of rules of the form
$f(q_1,\ldots,q_n)\to q$, where $f\in \mathcal{F}$ has arity $n$ and the $q_i$'s
and $q$ are in $Q$. A tree automaton $\A=(Q,\Delta,F)$ induces a relation on
$\T(\mathcal{F}\cup Q)$ (where elements of $Q$ are constant, denoted $\to_\A$ or
just $\to$, defined by $t\to_\A t^\prime$ if there exists a transition
$f(q_1,\ldots,q_n)\to q\in \Delta$ and $\alpha\in \pos(t)$ such that
$t^\prime=t[q]_\alpha$, $t(\alpha)=f$ and for every $1\leq i\leq n$,
$t(\alpha\cdot i)=q_i$. The reflexive transitive closure of $\to_\A$ is
denoted $\to_\A^*$. A term $t\in \T(\mathcal{F})$ is {\it accepted} by $\A$ if
there exists $q\in F$, such that $t\to_\A^* q$. An {\it run} $\rho$ for a
term $t\in \T(\mathcal{F})$ in $\A$ is a function from $\pos(t)$ into $Q$ such
that if $\alpha\in \pos(t)$ and $t(\alpha)$ has arity $n$, then
$t(\alpha)(\rho(\alpha\cdot 1),\ldots, \rho(\alpha\cdot n))\to \rho(\alpha)$
is in $\Delta$. An {\it accepting run} is a run satisfying
$\rho(\varepsilon)\in F$. It can be checked that a term $t$ is accepted by
$\A$ iff there exists an accepting run $\rho$ for $t$ and, more generally,
that $t\to_A^* q$ if there exists a run $\rho$ for $t$ in $\A$ such that
$\rho(\varepsilon)=q$. In this case we write $t\to_{\rho,\A}^* q$ or just
$t\to_{\rho}^* q$ if $\A$ is clear from the context.

A tree automaton with global equality and disequality constraints (TAGED for
short) is a tuple $(\A,R_1,R_2)$, where $\A=(Q,\Delta,F)$ is a tree
automaton and $R_1,R_2$ are binary relations over $Q$. The relation $R_1$ is
called the set of equality constraints and the relation $R_2$ the set of
disequality constraints. A term $t$ is accepted by $(\A,R_1,R_2)$ if there
exists a successful run $\rho$ for $t$ in $\A$ such that: if
$(\rho(\alpha),\rho(\beta))\in R_1$, then $t_{|\alpha}=t_{|\beta}$, and if
$(\rho(\alpha),\rho(\beta))\in R_2$, then $t_{|\alpha}\neq t_{|\beta}$.
For a ranked alphabet $\mathcal{F}$, 
let TAGED($k^\prime$,$k$) denote the class $(\A,R_1,R_2)$ of TAGED, where
$\A$ is a tree automaton over $\mathcal{F}$, $|R_1|\leq k^\prime$ and
 $|R_2|\leq k$.

\section{TAGED and the Hamiltonian Path Problem}
The paper focuses on proving the following theorem.

\begin{theorem}\label{thm:main}
The emptiness problem for TAGED(0,1) is NP-hard.
\end{theorem}

The proof of Theorem~\ref{thm:main} is a reduction to the  Hamiltonian Path
 Problem defined below.

\pb{Hamiltonian Graph Problem}{a directed finite graph $G=(V,E)$;}{$1$ if
  there exists a path in $G$ visiting each element of $V$ exactly once, $0$
  otherwise.}

The Hamiltonian Graph Problem is known to be NP-complete~\cite{garey}.
A path in a directed graph visiting  each vertex exactly once is called 
 a {\it Hamiltonian path}.
Before proving  Theorem~\ref{thm:main}, let us mention the following  direct important
consequence, which is the main result of the paper.

\begin{corollary}
For every fixed $k\geq 1$, and every fixed $k^\prime\geq 0$,
the emptiness problem for TAGED($k^\prime$,$k$) is NP-hard.
\end{corollary}

We have divided the proof of Theorem~\ref{thm:main} into a sequence of
lemmas. Lemma~\ref{lm:1}, below, is immediately obtained by a cardinality argument. 

\begin{lemma}\label{lm:1}
In a directed graph $G$ with $n$ vertices, there exists a Hamiltonian
path iff there is a path of length $n-1$ that does not visit the same
vertex  twice.
\end{lemma}


For any directed graph $G=(V,E)$, let $m_{G}$ denote the number of paths
of length $|V|-1$ in $G$.

\begin{lemma}\label{lemma:mG}
Let $G=(V,E)$ be a directed graph. One can compute $m_G$ in polynomial time
in the size of $G$.
\end{lemma}

\begin{proof}
Let us denote by $m_{G,k,u,v}$, for any $k\geq 1$, any $u\in V$ and any
$v\in V$, the number of paths of length $k$ from $u$ to $v$ in $G$. 
One has $m_{G,k+1,u,v}=\sum_{(u,u^\prime)\in E}m_{G,k,u^\prime,v}$.
Therefore, every $m_{G,k,u,v}$, for $k\leq |V|$, can be computed recursively
in polynomial time in $|V|$. Note that
$m_G=\sum_{u,v\in V}m_{G,|V|,u,v}$, concluding the proof. 
\end{proof}


Let $\mathcal{F}_1=\{f,g,A\}$, where $f$ has arity~$2$ and $g$ arity~$3$ and
$A$ is a constant. The next construction aims to build in polynomial 
time a tree automaton accepting a unique term having exactly $m$ leaves.


\begin{construction}
Let $m$ be a strictly positive integer and set
$\alpha_1\ldots\alpha_k$ the binary representation of $m$
($\alpha_1=1$ and $\alpha_i\in\{0,1\}$). Let $\A_m=(Q_1,\Delta_1,F_1)$
be the tree automaton over $\mathcal{F}_1$, where $Q_1=\{q_i \mid
0\leq i \leq k\}$, $F_1=\{q_k\}$ and $\Delta_1=\{A\to
q_1\}\cup\{f(q_{i},q_{i})\to q_{i+1}\mid 1\leq i \leq k-1\text{ and }
\alpha_{i+1}=0\}\cup\{g(q_{i},q_{i},q_1)\to q_{i+1}\mid 1\leq i \leq
k-1\text{ and } \alpha_{i+1}=1\}$.
\end{construction}

\begin{lemma}\label{lm:am}
The tree automaton $\A_m$ can be computed in polynomial time in $k$.
Moreover, $L(\A_m)$ is reduced to a single term having exactly $m$ leaves,
all labelled by $A$.
\end{lemma}

\begin{proof}
The proof is by induction on $k$. If $k=1$, then $m=1=\alpha_1$ (since
$m\neq 0$). In this case $Q_1=F_1=\{q_1\}$ and $\Delta_1=\{A\to q_1\}$;
therefore $L(\A_1)=\{A\}$ and the lemma result holds.   

Now assume that the lemma is true for a fixed $k\geq 1$. Let $2^{k+1}\leq m
< 2^{k+2}$ and set $m=\alpha_1\ldots\alpha_k\alpha_{k+1}$, the binary
representation of $m$. Two cases may arise:
\begin{itemize}
\item $\alpha_{k+1}=0$: In this case, by construction, the terms accepted by 
$\A_m$ are exactly the terms of the form $f(t_1,t_2)$, with 
$t_1\to^*_{\A_m} q_{k-1}$ and $t_2\to^*_{\A_m} q_{k-1}$. They correspond to the terms
$f(t_1,t_2)$, with $t_1,t_2\in L(\A_{\frac{m}{2}})$. By induction hypothesis, 
$L(\A_{\frac{m}{2}})$ is a singleton containing a unique term with
$\frac{m}{2}$ leaves, all labelled by $A$. It follows that $L(\A_m)$ accepts
a unique term with $2.\frac{m}{2}=m$ leaves, all labelled by $A$. 
\item  $\alpha_{k+1}=1$: Similarly, the terms accepted by 
$\A_m$ are exactly the terms of the form $g(t_1,t_2,A)$, with $t_1,t_2\in
L(\A_{\frac{m-1}{2}})$. By induction, it follows that $L(\A_m)$ accepts a
unique term with $1+2\cdot\frac{m-1}{2}=m$ leaves, all labelled by~$A$.
\end{itemize}
Therefore, the lemma result holds also for $k+1$, which concludes the proof. 
\end{proof}

Note that since $m_G\leq |V|^{|V|-1}$ 
the binary encoding of $m_G$ is of the size polynomial in $|V|$.
By Lemma~\ref{lemma:mG}, $m_G$ can be computed in polynomial time and $k$ is
polynomial in $|V|$. Therefore, the construction of $\A_{m_G}$ can be done in
polynomial time in $|V|$, proving the following lemma.

\begin{lemma}
Let $G$ be a directed graph satisfying $m_G\neq 0$. 
The tree automaton $\A_{m_G}$ can be computed in polynomial time.
\end{lemma}

The next construction is dedicated to a tree automaton $\P_{G}$ accepting 
terms encoding paths of length $|V|-1$. 

\begin{construction}
Let $G=(V,E)$ be a non empty directed graph and let $n=|V|-1$. Let
$\mathcal{F}_2=\{h\} \cup \{A_v\mid v\in V\}$, where $h$ is of arity~$2$ and
the $A_v$'s are constants. Let $\P_{G}=(Q_2,\Delta_2,F_2)$ be
the tree automaton over $\mathcal{F}_2)$, where
\item $Q_2=\{q_w^i \mid 0\leq i \leq n-1,\ w\in V\}$,
$F_2=\{q_{w}^{n-1}\mid w\in V\}$, and 
\begin{align*}
\Delta_2=\{A_w\to q_w^0 \mid w\in V\}
\cup \{h(q_v^0,q_w^i)\to q_v^{i+1}\mid 1\leq i \leq n-2,\ (w,v)\in V\} \ . \\
\end{align*}
\end{construction}

Note that the construction of $\P_G$ can be done in polynomial time. 
For a given graph $G=(V,E)$ and a given finite set $Q$, 
 an {\it $h$-term} on $Q$ is a term
either of the form $\beta_0$ or
$h(\beta_k,h(\beta_{k-1},h(\ldots,h({\beta_1},{\beta_0})\ldots)))$, where
$\beta_i\in \{A_v\mid v\in V\}\cup Q$. Such an $h$-term is denoted
$[\beta_k\beta_{k-1}\ldots\beta_0]_Q$. If $Q$ is clear from the context, the
index $Q$ is omitted.


\begin{lemma}\label{lm:p}
Let $G=(V,E)$ be a non empty directed graph. A term $t$ is
accepted by $\P_G$ iff there exists a path
$(w_0,w_1)(w_1,w_2)\ldots(w_{n-2},w_{n-1})$  in $G$ such that 
$t=[A_{w_{n-1}}A_{w_{n-2}}\ldots A_{w_1}A_{w_0}]_{Q_2}$.
\end{lemma}

\begin{proof}
If $t$ is accepted by $\P_G$, then there exists $w_{n-1}\in V$ such that 
 $t\to^*q_{w_{n-1}}^{n-1}$. Looking right-hand sides of the transitions, it
follows that there exists $w_{n-2}\in V$ such that 
$t\to^*h(q_{w_{n-1}}^0,q_{w_{n-2}}^{n-2})\to q_{w_{n-1}}^{n-1}$. The unique
rule with right-hand side $q_{w_{n-1}}^0$ is $A_{w_{n-1}}\to q_{w_{n-1}}^0$.
Therefore $t$ is of the form $t=h(A_{w_{n-1}},t^\prime)$ with 
$t^\prime\to^*q_{w_{n-2}}^{n-2}$ and $(w_{n-2},w_{n-1})\in E$. By a direct
induction on $n$, one has
$t=[A_{w_{n-1}}A_{w_{n-2}}\ldots A_{w_1}A_{w_0}],$
where $(w_0,w_1)(w_1,w_2)\ldots(w_{n-2},w_{n-1})$ is a path in $G$.

Conversely, assume that
$t=[A_{w_{n-1}}A_{w_{n-2}}\ldots A_{w_1}A_{w_0}]$
and that the sequence $(w_0,w_1)(w_1,w_2)\ldots(w_{n-2},w_{n-1})$ is a path in $G$.
For each $1\leq i \leq n-1$, let 
$t_i=[A_{w_{i}}A_{w_{i-1}}\ldots A_{w_1} A_{w_0}].$
One has $t_1=h(A_{w_1},A_{w_0})$, with $(w_0,w_1)\in E$. Therefore
$t_1\to^* q_{w_1}^1$. By a direct induction, one has
$t_i\to^*q_{w_i}^i$. Consequently $t_{n-1}\to^* q_{w_{n-1}}^{n-1}$.
It follows that $t_{n-1}$ is accepted by $\P_G$. It suffices to note that
$t_{n-1}=t$ to conclude the proof.
\end{proof}

The next construction designs a tree automaton $\C_{G}$ 
accepting 
terms of the form $[A_{w_k}A_{w_{k-1}}\ldots A_{w_1}A_{w_0}],$
where $k\geq 1$ and there exist $j\neq i$ such that 
$w_i=w_j$.  

\begin{construction}
Let $G=(V,E)$ be a non empty directed graph. Let
$\mathcal{F}_2=\{h\} \cup \{A_v\mid v\in V\}$, where $h$ has arity~$2$ and
the $A_v$'s are constants. 
Without loss of generality we assume that 
$0,1,f\notin V$. Let $\C_{G}=(Q_3,\Delta_3,F_2)$ be
the tree automaton over $\mathcal{F}_2$, where

$Q_3=\{p_w,p_w^\prime\mid w\in V\}\cup \{p_0,p_1,p_f\}$,
$F_3=\{p_f\}$, and 
\begin{align*}
\Delta_3=&\{A_w\to p_0,A_w\to p_w,A_w\to p_w^\prime\mid w\in
V\}\\&\cup\{h(p_0,p_0)\to p_1, h(p_0,p_1)\to p_1,h(p_w,p_0)\to
p_w^\prime\}\\ &\cup\{  h(p_w,p_w^\prime)\to p_f,
h(p_0,p_w^\prime)\to p_w^\prime, h(p_0,p_f)\to p_f,h(p_w,p_1)\to p_w^\prime \}.
\end{align*}
\end{construction}


\begin{lemma}\label{lm:p1}
Let $G=(V,E)$ be a non empty directed graph. For any term $t$, one has
$t \to^*_{\C_G} p_1$ iff $t=[A_{w_k}A_{w_{k-1}}\ldots A_{w_1}A_{w_0}]_{Q_3}$ 
with $k\geq 1$.
\end{lemma}

\begin{proof}
If $t=h(A_{w_k},h(A_{w_{k-1}},h(\ldots,h(A_{w_1},A_{w_0})\ldots))),$ 
then by a direct induction on $k$, and using the transitions $A_w\to p_0$ and
$h(p_0,p_1)\to p_1$, one has $t\to^* p_1$. 

Now, if $t\to^* p_1$, then the last transition used to reduce $t$ is
$h(p_0,p_1)\to p_1$. Therefore there exists $w\in V$ such that 
 $t=h(A_w,t^\prime)$ with $t^\prime\to^* p_1$. By a direct induction 
on the depth of $t$, one can conclude the proof. 
\end{proof}




\begin{lemma}\label{lm:pw}
Let $G=(V,E)$ be a non empty directed graph. For any term $t$, one has
$t \to^*_{\C_G} p_w^\prime$ iff $t$ is  of the form
$t=[A_{w_k}A_{w_{k-1}}\ldots A_{w_1}A_{w_0}]_{Q_3},$
where $k\geq 1$ and at least one of the $w_i$ is equal to $w$.
\end{lemma}

\begin{proof}
Let $t=[A_{w_k}A_{w_{k-1}}\ldots A_{w_1}A_{w_0}]_{Q_3}$ be a term such that
$w_i=w$, with $i\leq k$. If $i=0$, then $t\to [A_{w_k}A_{w_{k-1}}\ldots
  A_{w_1}p_w^\prime]\to^*p_w^\prime$ since $A_{w_0}=A_w\to p_w^\prime$. If 
$i=1$, then $t\to
[A_{w_k}A_{w_{k-1}}\ldots A_{w_1}p_0]\to^*[A_{w_k}A_{w_{k-1}}\ldots
    A_{w_2}p_w^\prime]$, using the transition $h(p_w,p_0)\to
p_w^\prime$. 
Now if $i\geq 2$, then, by Lemma~\ref{lm:p1}, one
  has $t\to^* [A_{w_k}A_{w_{k-1}}\ldots A_{w_i}p_1]$. Therefore
  $t\to^*[A_{w_k} A_{w_{k-1}}\ldots A_{w_{i+1}}p_wp_1]$. 
Since $[p_wp_1]\to
  p_w^\prime$, $t\to^*[A_{w_k} A_{w_{k-1}}\ldots
    A_{w_{i+1}}p_w^\prime]\to^*p_w^\prime$.

Conversely, if $t\to^* p_w^\prime$, we prove by induction on the depth
of $t$ that $t=[A_{w_k}A_{w_{k-1}}\ldots A_{w_1}A_{w_0}]$ with at least one
$i$ such that $w_i=w$. Assume now that the depth of $t$ is $n$. 
The four transitions having $p_w^\prime$ as right-hand
side are $A_w\to p_w^\prime$, $h(p_w,p_0)\to p_w^\prime$,
$h(p_0,p_w^\prime)\to p_w^\prime$ and $h(p_w,p_1)\to p_w^\prime$. If the
last transition used to reduce $t$ is $A_w\to p_w^\prime$, then $t=A_w$; $t$
is of the expected form. If the last transition used to reduce $t$ is
$h(p_w,p_1)\to p_w^\prime$, then $t=h(A_w,t^\prime)$. Using
Lemma~\ref{lm:p1}, $t$ is of the expected form. If the last transition used
to reduce $t$ is $h(p_w,p_0)\to p_w^\prime$, then there exists
$A_{w^\prime}$ such that $t=h(A_w,A_{w^\prime})$; $t$ is of the expected
form. If the last transition used to reduce $t$ is $h(p_0,p_w^\prime)\to
p_w^\prime$, then there exists $w^\prime$ and $t^\prime$ such that
$t=h(A_{w^\prime},t^\prime)$ and $t^\prime\to^* p_{w^\prime}$. By induction
hypothesis on $t^\prime$, $t$ is of the expected form, concluding the
induction and proving the lemma.
\end{proof}

\begin{lemma}\label{lm:C}
Let $G=(V,E)$ be a non empty directed graph. A term $t$ is
accepted by $\C_G$ iff it is of the form
$t=[A_{w_k}A_{w_{k-1}}\ldots A_{w_1}A_{w_0}]_{Q_3},$
where $k\geq 1$ and there exist $j\neq i$ such that 
$w_i=w_j$.  
\end{lemma}

\begin{proof}
Assume first
that $t=[A_{w_k}A_{w_{k-1}}\ldots A_{w_1}A_{w_0}]_{Q_3}$,
with $k\geq 2$ and there exist $j\neq i$ such that $w_i=w_j$.
If $j\geq 2$, one has
$$ t\to^*[A_{w_k}A_{w_{k-1}}\ldots A_{w_j}p_1] \to [A_{w_k}A_{w_{k-1}}\ldots
  p_{w_j}p_1] \to [A_{w_k}A_{w_{k-1}}\ldots A_{w_{j+1}}p_{w_j}^\prime].$$
If $j=1$, then $t\to [A_{w_k}A_{w_{k-1}}\ldots A_{w_1}p_{w_0}^\prime]
=[A_{w_k}A_{w_{k-1}}\ldots A_{w_{j+1}}p_{w_j}^\prime]$. If $j=0$, then $t\to
[A_{w_k}A_{w_{k-1}}\ldots A_{w_{1}}p_{w_0}^\prime]
=[A_{w_k}A_{w_{k-1}}\ldots A_{w_{j+1}}p_{w_j}^\prime]$. In every case one
has $t\to^*[A_{w_k}A_{w_{k-1}}\ldots A_{w_{j+1}}p_{w_j}^\prime]$. Moreover $
[A_{w_k}A_{w_{k-1}}\ldots A_{w_{j+1}}p_{w_j}^\prime]
\to^*[A_{w_k}A_{w_{k-1}}\ldots A_{w_{i}}p_{w_j}]$. Since $w_i=w_j$,
$$[A_{w_k}A_{w_{k-1}}\ldots A_{w_{i}}p_{w_j}]\to [A_{w_k}A_{w_{k-1}}\ldots
  A_{w_{i+1}}p_f].$$ It follows that $t$ is accepted by $\C_G$.

Conversely, assume now that $t\in L(\C_G)$. We prove by induction on
the depth of $t$ that it is of the form
$t=[A_{w_k}A_{w_{k-1}}\ldots A_{w_1}A_{w_0}],$ with $k\geq 2$ and such
that there exists $j\neq i$ satisfying $w_i=w_j$.

No constant is accepted by $\C_G$. If $t\in L(\C_G)$ has depth $2$, then
$t\to^* p_f$. The last transition used to reduce $t$ cannot be $h(p_0,p_f)\to
p_f$; otherwise $t$ would have a depth strictly greater than~$2$. It follows that there
exists $w$ such that $t\to h(p_w,p_w^\prime)$. Consequently, $t\to
h(A_w,p_w^\prime)$ since the unique transition having $p_w$ as right hand
side is $A_w\to p_w$. Now, since $t$ has depth~2, the unique possibility is
that $t=f(A_w,A_w)$. The property is therefore true for term of depth~2. Now
let $t$ be a term of depth $k-1$ belonging to $L(\C_G)$. There exists a
successful run $\rho$ such that $t\to^*_\rho p_f$. Therefore, either
$t\to_\rho^*h(p_0,p_f)$ or there exists $w_{k}$ such that
$t\to_\rho^*h(p_{w_{k}},p_{w_k}^\prime)$.
\begin{itemize}
\item If  $t\to_\rho^*h(p_0,p_f)$, then there exists $w_k$ such that 
$t=h(A_{w_k},t^\prime)$, with $t^\prime\in L(\C_G)$, and $t^\prime$ has depth~$k-1$. 
By induction on the depth, $t$ has the wanted form. 

\item If $t\to_\rho^* h(p_{w_k},p_{w_k}^\prime)$, then
  $t=h(A_{w_k},t^\prime)$ and $t^\prime\to^* w_k^\prime$. Using
  Lemma~\ref{lm:pw}, $t^\prime=[A_{w_{k-1}}A_{w_{k-2}}\ldots
    A_{w_1}A_{w_0}],$ where at least one of the $w_{i}$ ($i\leq k-1$) is
  equal to $w_k$, proving the induction and concluding the proof.
\end{itemize}
\end{proof}

\begin{lemma}\label{lm:bg}
Given a directed non empty graph $G=(V,E)$, one can compute in time polynomial
 in the size of $G$ a tree automaton $\B_G$ with a unique final state,
accepting exactly the set of terms of the form
$t=[A_{w_{|V|-2}}A_{w_{|V|-1}}\ldots A_{w_1}A_{w_0}]_\emptyset$,
such that $(w_0,w_1)\ldots (w_{|V|-1},w_{|V|-2})$ is a non Hamiltonian path
of $G$. 
\end{lemma}

\begin{proof}
The automata $\C_G$ -- checking that a vertex is visited twice --
and $\P_G$ -- checking the length of the path -- can both be computed in polynomial time.
Therefore, using the classical product construction, one can compute a tree
automaton accepting $L(\C_G)\cap L(\P_G)$ in polynomial time. Transforming
this automaton into an automaton with a unique final state can also be done in
polynomial time using classical $\varepsilon$-transition removal, proving
the lemma. The obtained automaton is $\B_G$.
\end{proof}

We can now give the last construction to prove the main result. 

\begin{construction}
Set $\B_G=(Q,\Delta,\{q_f\})$. Without loss of generality, one can assume
that $q_f=q_1$ and that $Q\cap Q_1=\{q_1\}$. We consider the automaton
$\D_G=(Q_4,\Delta_4,F_4)$ over $\mathcal{F}_1\cup\mathcal{F}_2$ defined by:
$Q_4=Q\cup Q_1$, $F_4=\{q_k\}$ and
$\Delta_4=(\Delta\cup\Delta_1)\setminus\{A\to q_1\}$.
\end{construction}

\begin{lemma}\label{lm:fin}
The TAGED $(\D_G,\emptyset,\{(q_1,q_1)\})$ can be constructed in polynomial time
in the size of $G$. Moreover, it accepts the empty language iff there exists
a Hamiltonian path in $G$. 
\end{lemma}

\begin{proof}
Using Lemma~\ref{lm:am}, the term accepted by $\D_G$ are those of the form
$C[t_1,\ldots,t_{m_G}]$, where $C[A,\ldots,A]$ is the unique term accepted
by $\A_{m_G}$ and each $t_i$ is accepted by $\B_G$. 
With the inequality constraint,  $(\D_G,\emptyset,\{(q_1,q_1)\})$ accepts an
empty language iff $|L(\B_G)|< m_G$. By Lemma~\ref{lm:bg}, $|L(\B_G)|$ is
the number of non Hamiltonian paths in $G$. Since $m_G$ is the number of
paths of length $|V|-1$ in $G$, using Lemma~\ref{lm:1},
 $L((\D_G,\emptyset,(q_1,q_1)))=\emptyset$ iff there exists a Hamiltonian
path of length $|V|-1$ in $G$.
\end{proof}

Theorem~\ref{thm:main} is a direct consequence of Lemma~\ref{lm:fin} and of
the polynomial time construction of $\D_G$.

\section{Conclusion}
In this paper we have proved that the emptiness problem for TAGED is NP-hard if
there is at least one negative constraint. It is known that the emptiness
problem for TAGED with only irreflexive disequality constraints is in
NEXPTIME~\cite{DBLP:journals/ijfcs/FiliotTT10}, and that it is NP-hard
-- by reduction of emptiness for DAG automata~\cite{charatonik1999automata}.
If there are only reflexive disequality constraints,
emptiness is known to be solvable in 3-EXPTIME~\cite{DBLP:journals/jar/CreusGG13}.
The gap between these bounds is large and deserves to be refined.

\bibliographystyle{plain}
\bibliography{tagedbib}

\end{document}